\documentclass[technote]{article}
\usepackage[margin=1in]{geometry}
\usepackage{amsmath}
\usepackage{amsthm}
\usepackage{amssymb}
\usepackage{graphicx}
\usepackage{subfig}
\usepackage{cite}
\usepackage{hyperref}
\usepackage{textcomp}
\usepackage[capitalise]{cleveref}

\newcommand{\norm}[1]{\left\lVert{#1}\right\rVert}

\newtheorem{theorem}{Theorem}[section]
\newtheorem{lemma}[theorem]{Lemma}

\title{A Comment on and Correction to: Opinion dynamics in the presence of increasing agreement pressure}
\author{
Christopher Griffin\thanks{C. Griffin is with the Communications, Information and Navigation Office, Applied Research Laboratory, University Park, PA 16802, E-mail: griffinch@psu.edu}
}

\begin{document}

\maketitle
\begin{abstract} We identify a counter-example to the consensus result given in [J. Semonsen et al. Opinion dynamics in the presence of increasing agreement pressure. \textit{IEEE Trans. Cyber.}, 49(4): 1270-1278, 2018]. We resolve the counter-example by replacing Lemma 5 in the given reference with a novel variation of the Banach Fixed Point theorem which explains both the numerical results in the reference and the counter-example(s) in this note, and provides a sufficient condition for consensus in systems with increasing peer-pressure. This work is relevant for other papers that have used the proof technique from Semonsen et al. and establishes the veracity of their claims assuming the new sufficient condition.
\end{abstract}

\section{Introduction}
In this technical note we correct and clarify a consensus result given in \cite{SGSR18}. This correction is relevant not only to the general literature but in particular to \cite{HZLR20}, which uses the same proof technique as \cite{SGSR18} and \cite{AVTM18,DMM18,CQLQ19,LLJD19,S19,ZLSR19,ZWAH20,AVTY20}, which cite \cite{SGSR18}. We show this proof method is incomplete due to the use of  a lemma drawn from outside the consensus literature. We provide a complete result and use this to prove a corrected version of Theorem 2 in \cite{SGSR18}. 

In \cite{SGSR18}, the authors study a consensus problem (see e.g., \cite{DeGroot74,Krause00,BHOT05,OFM07,BHT09,CFT12,JM14,MT14,Bindel2015248,7222408,7314902,7457653,7763766,7523325,7932950,7486092,7775015,7929330,7891027,7906607,7582369,7891010,7831431,7496995,7887685,7605476,7562413,7588177,7559818}) under an increasing peer-pressure function, which seems to drive system consensus. That paper assumes $N$ agents are arranged on a weighted graph with weighted adjacency matrix $\mathbf{A} \in \mathbb{R}^{N \times N}$ where self-weights are all $0$. Each agent has a time-varying state $x^{(i)} \in [0,1]$ (though any inputs in $\mathbb{R}$ would suffice) and the vector $\mathbf{x}_k \in \mathbb{R}^N$ is the vector of agent states at time $k$. Each agent also has a \textit{stubbornness} coefficient $s^{(i)}$ (also used in \cite{HZLR20}) and a preferred state ${x^{+^{(i)}}}$, which defines a fixed vector $\mathbf{x}^+ \in \mathbb{R}^{N\times N}$. Define a diagonal matrix $\mathbf{S}$ containing the $s^{(i)}$ and a diagonal matrix $\mathbf{D}$ of row-sums of $\mathbf{A}$. The the update function studied in \cite{SGSR18} is: 
\begin{equation}
\mathbf{x}_{(k)} = \left(\mathbf{S} + \rho_k\mathbf{D}\right)^{-1}\left(\mathbf{S}\mathbf{x}^{+} + \rho_k\mathbf{A}\mathbf{x}_{(k-1)}\right).
\label{eqn:Update}
\end{equation}
Here $\rho_k$ is a time-varying peer-pressure value. Let:
\begin{equation}
f_k(\mathbf{x}) = \left(\mathbf{S} + \rho_k\mathbf{D}\right)^{-1}\left(\mathbf{S}\mathbf{x}^{+} + \rho_k\mathbf{A}\mathbf{x}_{(k-1)}\right).
\end{equation}
In Lemma 3 of \cite{SGSR18} it is shown that $f_k(\mathbf{x})$ is a contraction with a fixed point:
\begin{displaymath}
\mathbf{x}_k^* = \left(\mathbf{S}+\rho_k\mathbf{L}\right)^{-1}\mathbf{S}\mathbf{x}^+,
\end{displaymath}
where: $\mathbf{L} = \mathbf{D} - \mathbf{A}$ is the Laplacian. It is noted in Theorem 1 of \cite{SGSR18} that: 
\begin{equation}
\lim_{k \to \infty}\mathbf{x}_k^* = \frac{\sum_{i = 1}^{N} s_{i}x_{i}^{+}}{\sum_{i = 1}^{N} s_{i}}{\bf 1}.
\label{eqn:ConvergencePoint}
\end{equation}
That is the fixed points of the individual contractions converge to the stubbornness weighted mean of the agents' preferred states. The authors state Lemma 5, taken directly from \cite{Lorentzen90,G91}:
\begin{lemma}[Theorem 1 of \cite{Lorentzen90} \& Theorem 2 of \cite{G91}] Let $\{f_n\}$ be a sequence of analytic contractions in a domain $D$ with $f_n(D) \subseteq E \subseteq D_0 \subseteq D$ for all $n$. Then $F_n = f_n \circ f_{n-1} \circ \cdots \circ f_1$ converges uniformly in $D_0$ and locally uniformly in $D$ to a constant function $F(z) = c \in E$. Furthermore, the fixed points of $f_{n}$ converge to the constant $c$.
\hfill\qed
\label{lem:Contraction2}
\end{lemma}
\noindent \cite{SGSR18} then uses this to argue (in Theorem 2) that when:
\begin{displaymath}
G_k(\mathbf{x}) = (f_k\circ \cdots \circ f_1)(\mathbf{x}),
\end{displaymath} 
if $\rho_k \to \infty$, then:
\begin{displaymath}
\lim_{k\to\infty} G_k(\mathbf{x}_0) = \mathbf{x}^*.
\end{displaymath}
In the next section, we show this is not a complete statement and that the system may fail to converge for certain choices of increasing $\rho_k$. The failure in this case is due to the use of \cref{lem:Contraction2} (Lemma 5 of \cite{SGSR18}), which appears not to be valid in this case. We then prove a variation of the Banach fixed-point theorem, which explains our example's failure to converge and provides a correct sufficient condition for convergence, thus completing Theorem 2 of \cite{SGSR18}.

\section{Counter-Example to Consensus}\label{sec:CounterExample}
Consider the simple graph $K_2$ with the following inputs:
\begin{displaymath}
\mathbf{A} = \begin{bmatrix}
0 & 1 \\
1 & 0
\end{bmatrix} \qquad \mathbf{S} = \mathbf{D} = \mathbf{I}_2
\end{displaymath}
Let the initial condition and preferred agent states be given by $\mathbf{x}^+ = \langle{0.1,0.5}\rangle$. Assume we define the exponentially increasing peer-pressure function:
\begin{displaymath}
\rho_k = 2^{\sqrt{k}},
\end{displaymath}
which provides some numerical stability (i.e., does not blow up too quickly) but also shows exponential growth. Simulation of \cref{eqn:Update} shows the system does not converge to the expected $\mathbf{x^*} = \langle{3,3}\rangle$ as given by \cref{eqn:ConvergencePoint}, but instead oscillates about this point indefinitely. This is shown in \cref{fig:Oscillation}.
\begin{figure}[htbp]
\centering
\subfloat[Initial Behavior]{
\includegraphics[width=0.4\textwidth]{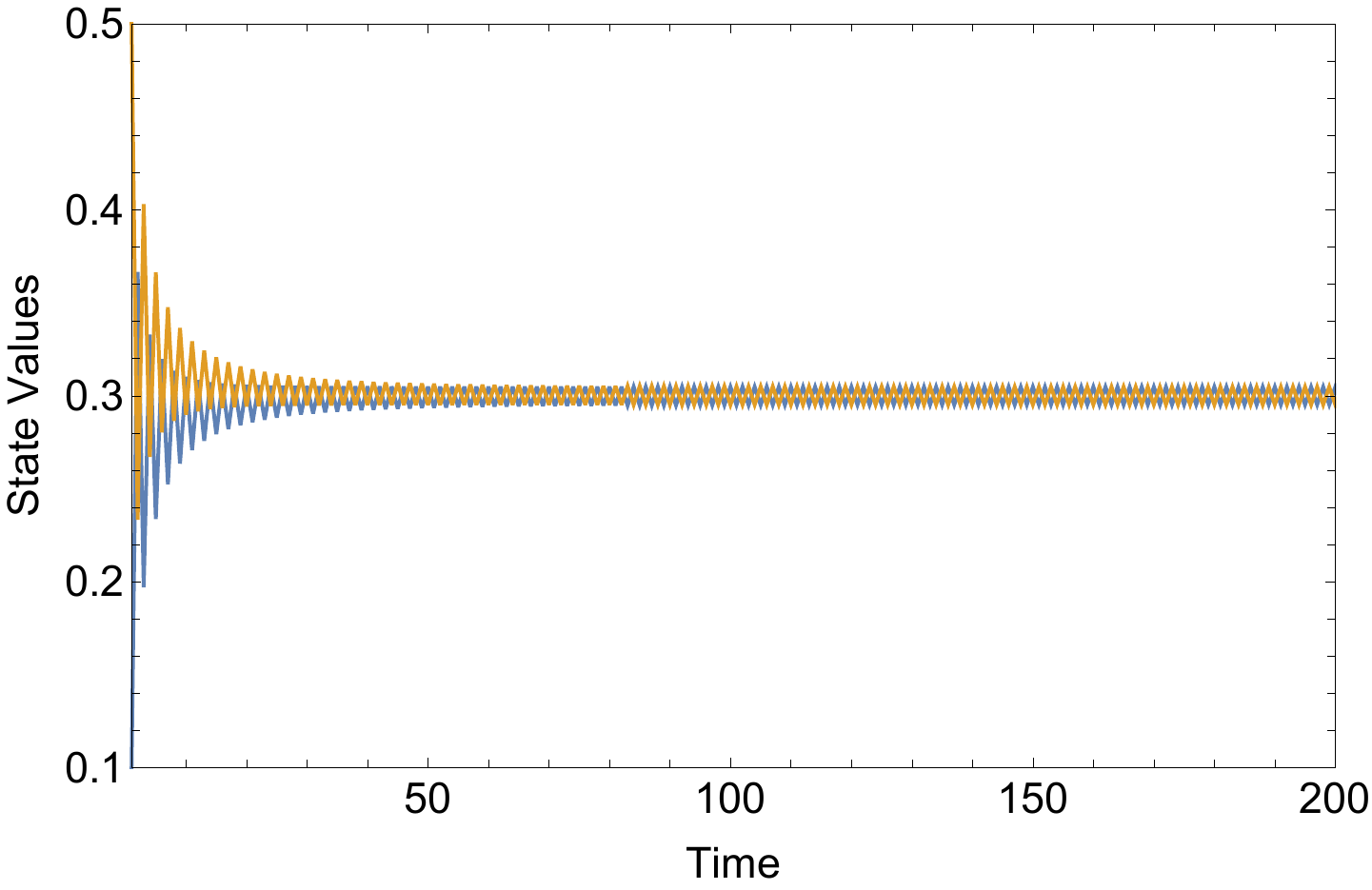}} \\
\subfloat[Long Run Oscillation]{
\includegraphics[width=0.4\textwidth]{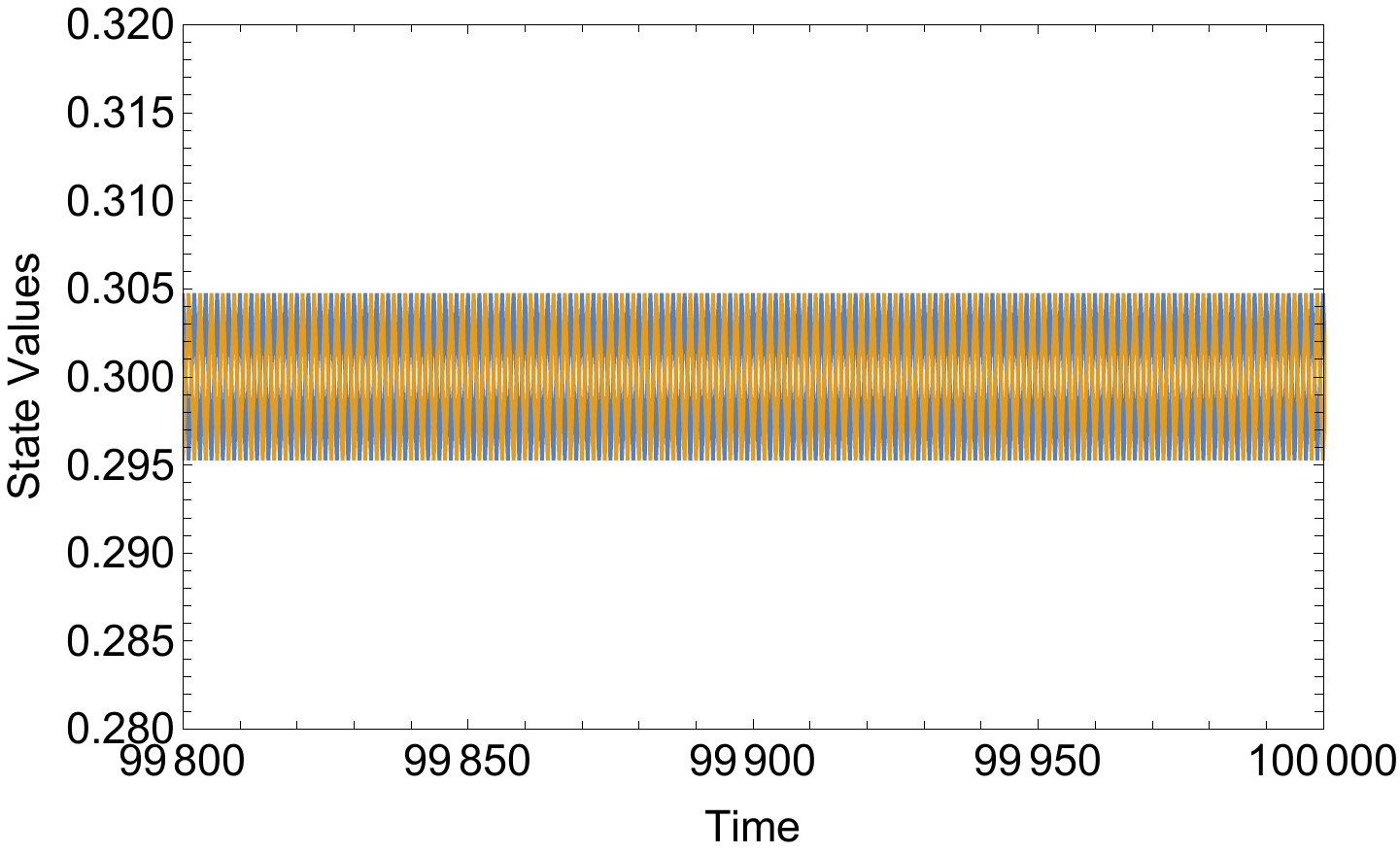}} 
\caption{Initially fast convergence around the mean point $\mathbf{x^*} = \langle{3,3}\rangle$ slows and becomes oscillation, showing neutral stability, rather than asymptotic stability.}
\label{fig:Oscillation}
\end{figure}
However, if we replace the peer-pressure function with:
\begin{displaymath}
\rho_k = k,
\end{displaymath}
then we see the system converges as expected from Theorem 1 of \cite{SGSR18}.
\begin{figure}[htbp]
\centering
\subfloat[Initial Behavior]{
\includegraphics[width=0.4\textwidth]{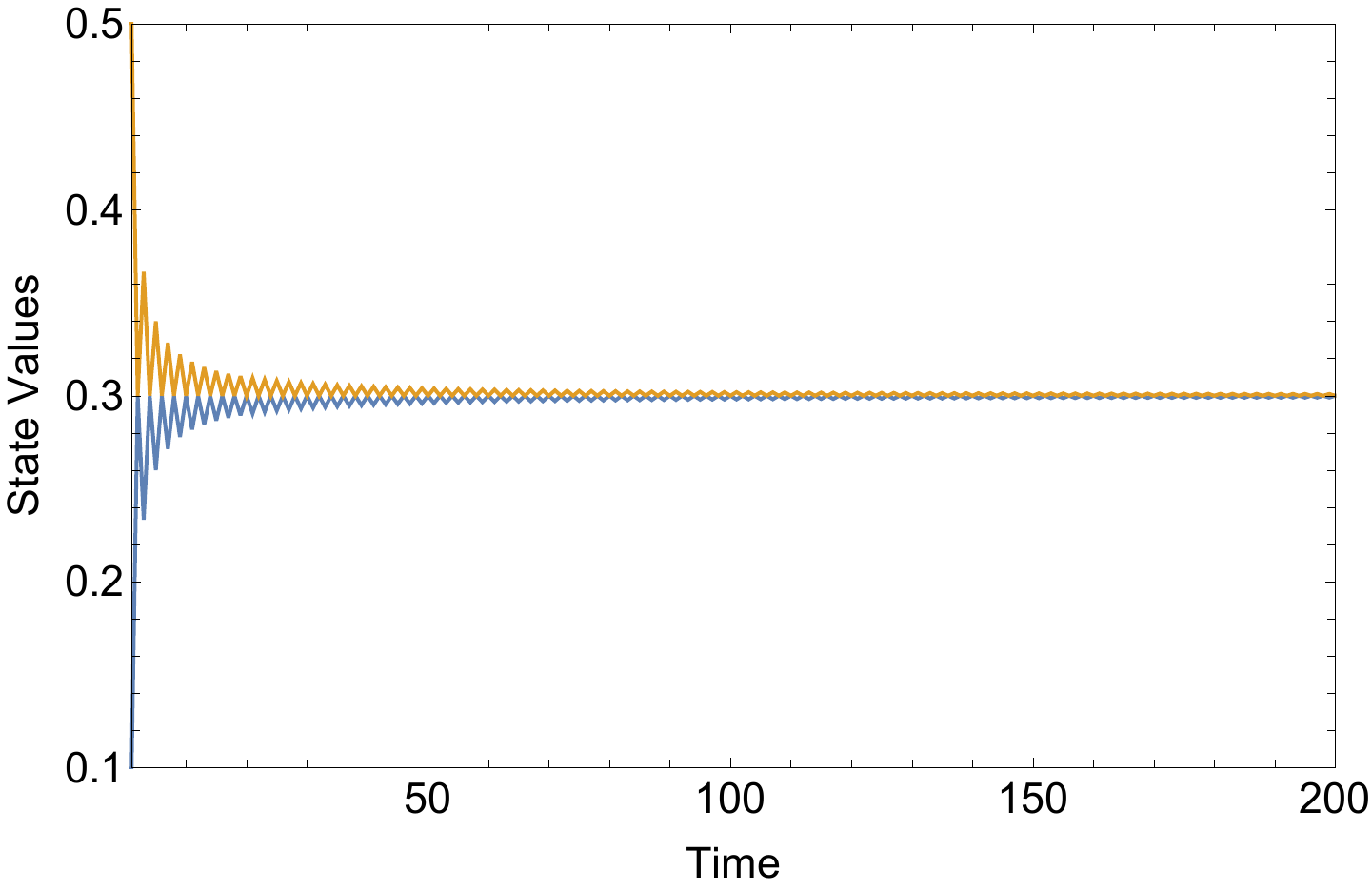}} \\
\subfloat[Long Run Oscillation]{
\includegraphics[width=0.4\textwidth]{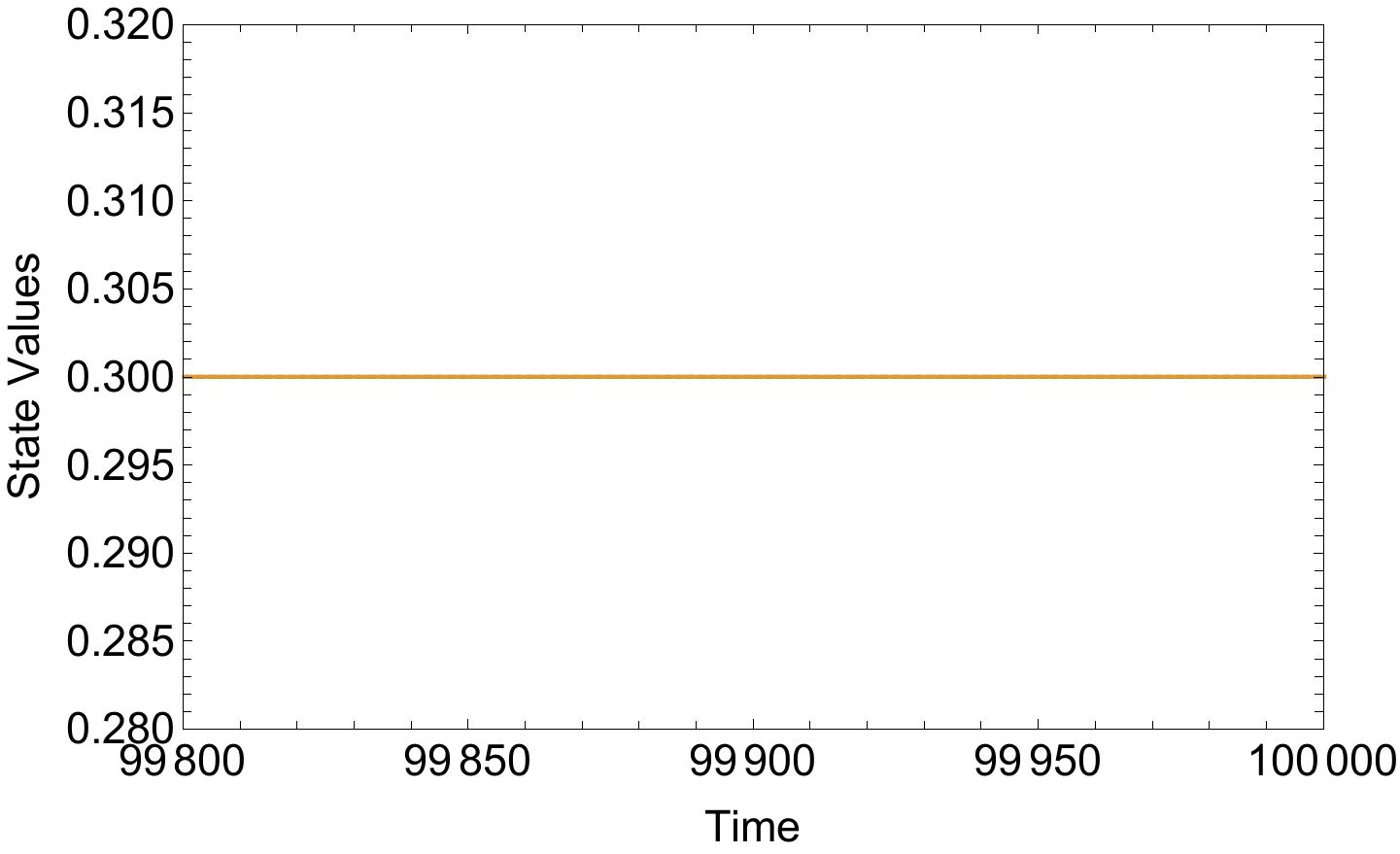}} 
\caption{Asymptotic convergence to the point $\mathbf{x^*} = \langle{3,3}\rangle$ is illustrated. As noted in \cite{SGSR18} this convergence is linear.}
\label{fig:Convergence}
\end{figure}
It is clear from this example that the issue with Theorem 1 is not the statement of the theorem, but its lack of qualification on the growth of $\rho_k$. This stems directly from Lemma 5 of \cite{SGSR18} (or Theorem 1 of \cite{Lorentzen90} \& Theorem 2 of \cite{G91}), which also does not qualify the analytic contraction to be used. However, it is clear that this pathology is an example of an even easier example. 

Consider the family of functions $f:\mathbb{R} \to \mathbb{R}$: 
\begin{equation}
f_k(x) = \left(1 - \frac{1}{10^n}\right) x.
\label{eqn:Ex1}
\end{equation}
Each $f_k(x)$ has a fixed point $x_k^* = 0$, thus the fixed points $x_k^*$ converge to $x^* = 0$ (tautologically). Moreover, each function contracts any interval containing $x = 0$ into itself. However, computing:
\begin{displaymath}
G_k(x) = (f_k\circ f_{k-1} \circ \cdots \circ f_1)(x) = \left(\prod_{i=1}^k\left(1 - \frac{1}{10^n}\right)\right)x,
\end{displaymath}  
we see that:
\begin{displaymath}
\lim_{k\to\infty} G_k(x) = x\cdot \lim_{k\to\infty}\prod_{n=1}^k\left(1 - \frac{1}{10^n}\right) = x \cdot \phi\left(\tfrac{1}{10}\right).
\end{displaymath}
Here $\phi(\cdot)$ is Euler's function derived from the q-Pochhammer symbol. We note that: 
\begin{displaymath}
\phi\left(\tfrac{1}{10}\right) \approx 0.89001
\end{displaymath}
Therefore, for $x \neq 0$, $G_\infty(x) \approx 0.89001x$, rather than $0$ as would be expected from the statement of Lemma 5 of \cite{SGSR18}. By contrast, if we consider the family of functions:
\begin{equation}
f_k(x) = \frac{k-1}{k} x,
\label{eqn:Ex2}
\end{equation}
then:
\begin{displaymath}
G_k(x)=\left(\prod_{n=1}^k\frac{n-1}{n}\right)x
\end{displaymath}
and
\begin{displaymath}
\lim_{k\to\infty} G_k(x) = x\cdot \lim_{k\to\infty}\left(\prod_{n=1}^k\frac{n-1}{n}\right) = 0,
\end{displaymath} 
as expected. We note this counter-example can be extended to the complex plane (the domain used in \cref{lem:Contraction2}). In this note, we will not consider the question of whether the statement of \cref{lem:Contraction2} needs to be clarified, however it may be that this result should be used with caution for future consensus results.

In the next section, we construct a variation of the Banach fixed point theorem that handles the conditions set forth in \cite{SGSR18} and predicts the non-convergence of the counter-example using the intuition provided by the simpler cases.

\section{A Convergence Theorem}
\begin{theorem} Let $\{f_i\}_{i=1}^\infty$ be a family of mappings on a Banach space $X$ with norm $\norm{\cdot}$ so that each mapping $f_i$ has a unique fixed point $\mathbf{x}_i^*$ satisfying the property:
\begin{displaymath}
\norm{f_i(\mathbf{x}) - \mathbf{x}_i^*} \leq \alpha_i \norm{\mathbf{x} - \mathbf{x}_i^*} \qquad \forall \mathbf{x} \in X,
\end{displaymath}
where $0 \leq \alpha_i < 1$ for all $i$. Furthermore, suppose that:
\begin{displaymath}
    \lim_{i\to\infty} \mathbf{x}_i^* = \mathbf{x}^* \in X
\end{displaymath}
Then, if $\mathbf{x}_0 \in X$ and
\begin{displaymath}
\lim_{N \to \infty} \prod_{i=1}^N \alpha_i = 0,
\end{displaymath}
then:
\begin{displaymath}
  \lim_{k \to \infty} (f_k\circ f_{k-1} \circ \cdots \circ f_1)(\mathbf{x}_0) = \mathbf{x}^*
\end{displaymath}
\label{thm:Contraction}
\end{theorem} 
\begin{proof}
Define:
\begin{displaymath}
G_k = f_k \circ f_{k-1} \circ \cdots \circ f_1.
\end{displaymath}
By assumption:
\begin{displaymath}
\norm{(f_{k+1}\circ G_k)(\mathbf{x}_0) - \mathbf{x}_{k+1}^*} \leq \alpha_{k+1}\norm{G_k(\mathbf{x}_0) - \mathbf{x}_{k+1}^*} =
\alpha_{k+1}\norm{(f_k\circ G_{k-1})(\mathbf{x}_0) - \mathbf{x}_{k+1}^*}.
\end{displaymath}
Applying the triangle inequality to the last term we have:
\begin{displaymath}
\norm{(f_{k+1}\circ G_k)(\mathbf{x}_0) - \mathbf{x}_{k+1}^*} \leq 
\alpha_{k+1}\left( \norm{(f_k\circ G_{k-1})(\mathbf{x}_0) - \mathbf{x}_{k}^*} + \norm{\mathbf{x}_{k+1}^* - \mathbf{x}_{k}^*} \right).
\end{displaymath}
Applying similar logic, we see that:
\begin{displaymath}
\norm{(f_{k+1}\circ G_k)(\mathbf{x}_0) - \mathbf{x}_{k+1}^*} \leq 
\alpha_{k+1}\alpha_{k}\norm{G_{k-1}(\mathbf{x}_0) - \mathbf{x}_{k}^*} +
\alpha_{k+1} \norm{\mathbf{x}_{k+1}^* - \mathbf{x}_{k}^*}
\end{displaymath}
Repeating this argument, we see:
\begin{displaymath}
\norm{(f_{k+1}\circ G_k)(\mathbf{x}_0) - \mathbf{x}_{k+1}^*} \leq\\ 
\alpha_{k+1}\alpha_{k}\alpha_{k-1}\norm{G_{k-2}(\mathbf{x}_0) - \mathbf{x}_{k-1}^*} +\\
\alpha_{k+1}\alpha_k \norm{\mathbf{x}_{k}^* - \mathbf{x}_{k-1}^*} + 
\alpha_{k+1} \norm{\mathbf{x}_{k+1}^* - \mathbf{x}_{k}^*}
\end{displaymath}
We can continue in this way until we see that:
\begin{equation}
\norm{(f_{k+1}\circ G_k)(\mathbf{x}_0) - \mathbf{x}_{k+1}^*} \leq 
\left(\prod_{i=1}^{k+1}\alpha_i\right) \norm{f_1(\mathbf{x}_0) - \mathbf{x}_1^*} + \\\sum_{j=1}^k\left(\prod_{i={j+1}}^{k+1}\alpha_i\right)\norm{\mathbf{x}_{j+1}^* - \mathbf{x}_j^*}.
\label{eqn:MainEqn}
\end{equation}
By assumption of the theorem, the fixed points $\mathbf{x}_i^*$ converge and therefore for any $\epsilon > 0$ there is an $N > 0$ so that:
\begin{equation}
\norm{\mathbf{x}_{N}^* - \mathbf{x}_{N+1}^*} \leq \epsilon.
\end{equation}

Suppose we are given an $\epsilon > 0$, choose $N$ so that:
\begin{enumerate}
\item 
\begin{displaymath}
\left(\prod_{i=1}^{k+1}\alpha_i\right) \norm{f_1(\mathbf{x}_0) - \mathbf{x}_1^*} < \frac{\epsilon}{2(N+1)}.
\end{displaymath}
This is possible since we assume:
\begin{equation}
\lim_{N \to \infty} \prod_{i=1}^N \alpha_i = 0,
\label{eqn:rholim}
\end{equation}
and
\begin{equation}
\lim_{N \to \infty} \norm{\mathbf{x}_{N+1} - \mathbf{x}_N} = 0.
\label{eqn:xlim}
\end{equation}

\item For each $k$, 
\begin{displaymath}
\left(\prod_{i={j+1}}^{N+1}\alpha_i\right)\norm{\mathbf{x}_{j+1}^* - \mathbf{x}_j^*} < \frac{\epsilon}{2(N+1)}.
\end{displaymath}
This is possible because of the combination of \cref{eqn:xlim,eqn:rholim}. 

\item 
\begin{displaymath}
\norm{\mathbf{x}_{N+1}^* - \mathbf{x}^*} < \frac{\epsilon}{2}
\end{displaymath}
\end{enumerate}
Then from \cref{eqn:MainEqn} we have:
\begin{displaymath}
\norm{G_{n+1}(\mathbf{x}_0) - \mathbf{x}_{n+1}^*} < \frac{\epsilon}{2}.
\end{displaymath}
By one more application of the triangle inequality, we have:
\begin{displaymath}
\norm{G_{n+1}(\mathbf{x}_0) - \mathbf{x}^*} \leq 
\norm{\mathbf{x}_{N+1}^* - \mathbf{x}^*} + \norm{G_{n+1}(\mathbf{x}_0) - \mathbf{x}_{n+1}^*} < \epsilon
\end{displaymath}
This completes the proof.
\end{proof}
This theorem explains, as special cases, the families of functions defined in \cref{eqn:Ex1} and \cref{eqn:Ex2}. In the next section, we use it to resolve the counter-example discussed in \cref{sec:CounterExample}. 
 
\section{Resolution of the Counter Example}
Returning to the example in \cref{sec:CounterExample}, let $\mathbf{x}^+ = \langle{a,b}\rangle$, thus generalizing the initial condition. In this case, \cref{eqn:Update} can be written as:
\begin{equation}
f_k(\mathbf{x}) = \begin{bmatrix}
\frac{(\rho +1) (a+\rho  x_2)}{\rho ^2+2 \rho +1} \\
 \frac{(\rho +1) (b+\rho  x_1)}{\rho ^2+2 \rho +1}
\end{bmatrix}.
\end{equation}
For each $f_k(\mathbf{x})$ the explicit fixed point is given by:
\begin{equation}
\mathbf{x}_k^* = \begin{bmatrix}
x_1^*\\
x_2^*
\end{bmatrix} = \begin{bmatrix}
\frac{a + \rho(a+b) }{2 \rho +1}\\
\frac{b+ \rho(a+b)}{2 \rho +1}
\end{bmatrix}.
\end{equation}
Since $\langle{s_1,s_2}\rangle = \langle{1,1}\rangle$, we see that:
\begin{displaymath}
\lim_{k \to \infty} \mathbf{x}_k^* = \mathbf{x}^* = \begin{bmatrix}
\tfrac{1}{2}\\\tfrac{1}{2}
\end{bmatrix},
\end{displaymath}
as expected.

Because this example is particularly simple, we can compute (see \cref{sec:Norms}):
\begin{equation}
\norm{F(\mathbf{x}) - \mathbf{x}_k^*}^2 = \left(\frac{\rho_k}{1+\rho_k}\right)^2\norm{\mathbf{x} - \mathbf{x}_k^*}^2,
\label{eqn:FixedPointNorm}
\end{equation}
for any $\mathbf{x}$. In \cref{thm:Contraction}, we now have:
\begin{displaymath}
\alpha_k = \frac{\rho_k}{1+\rho_k}.
\end{displaymath}
When $\rho_k = 2^{\sqrt{k}}$, then:
\begin{displaymath}
\prod_{k=1}^\infty \alpha_k \approx 0.0310128 > 0,
\end{displaymath}
which implies (as expected) that the system may not converge to the fixed point (see \cref{fig:Oscillation}). On the other hand, when $\rho_k = k$,
\begin{displaymath}
\prod_{k=1}^\infty \alpha_k = 0
\end{displaymath}
ensuring that the system will converge to the weighted average of $\mathbf{x}^+$. Thus we can correct Theorem 2 of \cite{SGSR18} to read:
\begin{theorem}[Clarification of Theorem 2 of \cite{SGSR18}] Let 
\begin{displaymath}
f_k(\mathbf{x}) = \left(\mathbf{S} + \rho^{(k)}\mathbf{D}\right)^{-1}\left(\mathbf{S}\mathbf{x}^{+} + \rho^{(k)}\mathbf{A}\mathbf{x}\right).
\end{displaymath}
Define:
\begin{displaymath}
\mathbf{x}_{k} = f_k(\mathbf{x}_{k-1}),
\end{displaymath}
with $\mathbf{x}_0$ given (and assumed to be $\mathbf{x}^+$). If $\rho_k \to \infty$ and the resulting contraction constants $\alpha_k$ of $f_k$ satisfy:
\begin{displaymath}
\prod_{k=1}^\infty \alpha_k = 0,
\end{displaymath}
then
\begin{displaymath}
\lim_{k \to \infty} \mathbf{x}_k = \frac{\sum_i s_i\mathbf{x}_{k_i}}{\sum_i s_i}.
\end{displaymath}\hfill\qed
\label{thm:FixedThm}
\end{theorem}
As a final remark, we note that this result could be anticipated from the convergence rate analysis in \cite{SGSR18}, which shows that \cref{eqn:Update} is an instance of gradient descent. Convergence guarantees for such an algorithm require satisfaction of the Wolfe Conditions and several pathological examples exist in which the step length (governed by $\rho_k$) is improperly defined leading to oscillation in gradient descent (see \cite{Bert99}). Thus, one might view \cref{thm:FixedThm} as a specialized sufficient condition on step length in this gradient descent. 

\section{Conclusions}
In this technical note, we corrected and clarified Theorem 2 of \cite{SGSR18}. This correction is important because the proof method has been used by other authors \cite{HZLR20}. The correction is based on replacing a lemma (Lemma 5) used in \cite{SGSR18} with a new variation on the Banach Fixed Point theorem. The modified theorem(s) now ensure results in \cite{SGSR18} and \cite{HZLR20} can be used for development of consensus systems or for their further study.

\section*{Acknowledgement}
C.G. was supported in part by the National Science Foundation under grant CMMI-1932991.

\appendix
\section{Computation of the Norms}\label{sec:Norms}
Computing directly we have:
\begin{displaymath}
\norm{f_k(\mathbf{x}) - \mathbf{x}_k^*}^2 = 
\left(\frac{a+\rho  x_2}{\rho +1}-\frac{a + \rho(a+b)}{2 \rho
   +1}\right)^2+\\
\left(\frac{b+\rho  x_1}{\rho +1}-\frac{b+ \rho(a+b)}{2 \rho
   +1}\right)^2
\end{displaymath}
We also have:
\begin{displaymath}
\norm{\mathbf{x} - \mathbf{x}_k^*}^2 = 
\left(x_1-\frac{\rho  (a+b)+a}{2 \rho +1}\right)^2+\\
\left(x_2-\frac{\rho  (a+b)+b}{2
   \rho +1}\right)^2.
\end{displaymath}
Dividing these expressions into each other and simplifying\footnote{Using Mathematica\texttrademark.} yields \cref{eqn:FixedPointNorm}.

\bibliographystyle{unsrt}
\bibliography{ConvergenceClarification}
 
\end{document}